\documentclass[12pt]{article}

\RequirePackage{amsthm,amsmath,amsfonts,amssymb}
\RequirePackage[authoryear]{natbib} 
\usepackage{xr-hyper}
\RequirePackage[colorlinks,citecolor=blue,urlcolor=blue]{hyperref}
\RequirePackage{graphicx}

\newcommand{\blind}{0}

\usepackage{geometry}

\def\spacingset#1{\renewcommand{\baselinestretch}%
	{#1}\small\normalsize} \spacingset{1}

\numberwithin{equation}{section}
\numberwithin{figure}{section}

\makeatletter
\def\namedlabel#1#2{\begingroup
	#2%
	\def\@currentlabel{#2}%
	\phantomsection\label{#1}\endgroup
}
\makeatother

\newcommand*{\fullref}[1]{\hyperref[{#1}]{\autoref*{#1} of \nameref*{#1}}}

\usepackage{float}
\usepackage{bm}
\usepackage{dsfont}
\usepackage{amssymb}
\usepackage{amstext}
\usepackage{graphics}
\usepackage{graphicx}
\usepackage{amsthm}
\usepackage{amsmath}
\usepackage{mathtools}
\usepackage{enumitem}
\usepackage{color}

\usepackage{booktabs}
\usepackage{longtable}
\usepackage{array}
\usepackage{multirow}
\usepackage{wrapfig}
\usepackage{float}
\usepackage{colortbl}
\usepackage{pdflscape}
\usepackage{tabu}
\usepackage{threeparttable}
\usepackage{threeparttablex}
\usepackage[normalem]{ulem}
\usepackage{makecell}
\usepackage{xcolor}

\usepackage{natbib}

\graphicspath{{figure/}}

\theoremstyle{plain}
\newtheorem{thm}{Theorem}[section]

\theoremstyle{definition}

\theoremstyle{remark}

\newtheorem{ps}{Proposition}[section]
\newtheorem{corollary}{Corollary}[section]

\bibpunct{(}{)}{;}{a}{}{,}
\AtBeginDocument{}
\makeatletter
\newsavebox{\mybox}\newsavebox{\mysim}
\newcommand{\distas}[1]{%
  \savebox{\mybox}{\hbox{\kern3pt$\scriptstyle#1$\kern3pt}}%
  \savebox{\mysim}{\hbox{$\sim$}}%
  \mathbin{\overset{\text{#1}}{\kern\z@\resizebox{\wd\mybox}{\ht\mysim}{$\sim$}}}%
}
\makeatother
\newcommand{\defas}{\overset{\triangle}{=}}

\newcommand{\inner}[2]{\left\langle #1, #2 \right\rangle}

\newcommand*{\cbrk}[1]{\left\{ #1 \right\}} 

\DeclareMathOperator*{\argmax}{arg\,max}

\newcommand{\toinf}{\rightarrow \infty}
\newcommand{\tozero}{\rightarrow 0}

\newcommand{\calP}{\mathcal{P}}
\newcommand{\calG}{\mathcal{G}}
\newcommand{\calT}{\mathcal{T}}

\newcommand{\dsR}{\mathds{R}}

\newcommand{\vphi}{\varphi}

\newcommand*{\tl}[1]{\tilde{#1}}
\newcommand*{\h}[1]{\hat{#1}}
\newcommand*{\ck}[1]{\check{#1}}
\newcommand*{\rng}[1]{\mathring{#1}}

\newcommand*{\hmu}{\hat{\mu}}
\newcommand*{\hl}{\hat{\lambda}}
\newcommand*{\hphi}{\hat{\varphi}}

\newcommand{\rnk}{\operatorname{rank}}

\newcommand{\spn}{\operatorname{span}}

\renewcommand{\Pr}{P}
\newcommand{\E}{E}

\newcommand{\Var}{\operatorname{Var}}
\newcommand{\var}{\operatorname{var}}

\newcommand{\cov}{\operatorname{cov}}
\newcommand{\ave}{\text{ave}}

\DeclarePairedDelimiterX{\infdivx}[2]{(}{)}{%
	#1\;\delimsize\|\;#2%
}

\newcommand*{\kldiv}[2]{D\infdivx{#1}{#2}}

\newcommand*{\inv}{^{-1}}

\newcommand*{\invsqr}{^{-1/2}}

\usepackage{amsthm}
\theoremstyle{compact.definition}

\makeatletter
\def\th@newremark{\th@remark\thm@headfont{\bfseries}} 
\makeatletter
\theoremstyle{newremark}

\theoremstyle{compact.remark}

\newcommand{\clog}{\operatorname{\psi}}

\newcommand{\blup}{\operatorname{BLUP}}
\newcommand{\aic}{\operatorname{AIC}}
\newcommand{\cxEtrp}{H}
\newcommand{\knl}{\kappa}
\newcommand{\mco}{\xi} 

\newcommand{\normp}[2][2]{\left\lVert#2\right\rVert_#1}
\newcommand{\norminf}[1]{\lVert#1\rVert_\infty}

\DeclarePairedDelimiterX{\abs}[1]{\lvert}{\rvert}{#1}
\newcommand{\size}[1]{\left\lvert#1\right\rvert}

\DeclareMathAlphabet{\calbnx}{U}{BOONDOX-calo}{m}{n}
\SetMathAlphabet{\calbnx}{bold}{U}{BOONDOX-calo}{b}{n}
\DeclareMathAlphabet{\calbnxb}{U}{BOONDOX-calo}{b}{n}

\graphicspath{{figure/}}

\title{Nonparametric Estimation of Repeated Densities with Heterogeneous Sample Sizes}	

\if0\blind
{
	\author{
		Jiaming Qiu%
		\and 
		Xiongtao Dai%
		\and 
		Zhengyuan Zhu \\
		Department of Statistics, Iowa State University
	}
} \fi
\if1\blind
{
	\author{}
} \fi

\date{}

\begin{document}

\maketitle

\begin{abstract}
We consider the estimation of densities in multiple subpopulations, where the available sample size in each subpopulation greatly varies. 
This problem occurs in epidemiology, for example, where different diseases may share similar pathogenic mechanism but differ in their prevalence. 
Without specifying a parametric form, our proposed method pools information from the population and estimate the density in each subpopulation in a data-driven fashion.
Drawing from functional data analysis, 
low-dimensional approximating density families in the form of exponential families are constructed from the principal modes of variation in the log-densities.
Subpopulation densities are subsequently fitted in the approximating families based on likelihood principles and shrinkage. 
The approximating families increase in their flexibility as the number of components increases and can approximate arbitrary infinite-dimensional densities.
We also derive convergence results of the density estimates with discrete observations. 
The proposed methods are shown to be interpretable and efficient in simulation as well as applications to electronic medical record and rainfall data.
\end{abstract}

\noindent%
\vfill


\section{Introduction}

Density estimation is one of the most fundamental tasks in statistics and machine learning. 
We consider the problem of repeated density estimation, namely, to estimate the distributions of a variable in multiple subpopulations with similar nature.
For example, the distributions of interest can be the age distributions of patients with different diseases, or the precipitation distributions in multiple regions. 
Repeated density functions can be modeled as the realizations of a random density element that is subject to the positivity and unit integral constraints. 
Moreover, density functions are in practice unobserved and are accessed through discrete samples, rendering additional difficulties for data analysis.
The primary objective of this work is to address these difficulties and derive reliable yet flexible density estimates for the subpopulations given discrete observations, in a fully data-driven and efficient manner.

Most of the existing work consider random densities that are either \emph{completely} observed or are reconstructed from \emph{densely} sampled observations. 
This type of density objects was first analyzed by \cite{knei:01} by directly applying functional principal component analysis (FPCA), a technique that has later been applied to analyze time-varying densities \citep{huyn:11,tsay:16,chu:18}. 
Recognizing the constraints in density functions, \cite{deli:11} and \cite{pete:16} proposed to analyze densities by utilizing one-to-one transformations to map them to an unconstrained space, where the transformed densities are then represented using FPCA and subsequently back transformed into densities.
Distinct geometries are generated on the space of densities by different transformations, such as the log-hazard and log-quantile density transformations \citep{pete:16,kim:20}, log transformations for analyzing densities as compositional data \citep{egoz:06,hron:16}; and the square root transformation with the Riemannian logarithm map \citep{sriv:07,dai:20:pp} which generates a Hilbert sphere geometry.

We instead consider fitting \emph{sparsely} observed random densities with as few as a handful of realizations available. 
Subpopulations with small sample sizes arise in a number of possible situations when the sampling mechanism dictates heterogeneous sample sizes. 
A first situation is when the availability and cost of observation highly varies among subpopulations \citep{wake:99}, with some groups easy but the rest hard to reach \citep{bone:14}.
This scenario is illustrated by the Medical Information Mart for Intensive Care (MIMIC) data \citep{john:16} in our first data application, where patient records are ample for common diseases such as diabetes and coronary heart disease but scarce for uncommon conditions.
A second generating mechanism is that the data collection process had been established for different duration in different subpopulations \citep{fith:14}, which is illustrated by our second data application considering rainfall recorded by weather stations around the Haihe River basin in northern China that differ in the number of years when data are available.

The literature on sparsely observed random densities is scant.
A closely related topic is replicated point processes, of which the intensity functions can be normalized into densities. 
\cite{pana:16} considered time-warping densely sampled Poisson point processes under a transportation geometry.
Given sparsely sampled replicated point processes, \cite{wu:13} proposed to pool information over multiple realizations and directly model the intensity functions using FPCA.
\cite{gerv:16} and \cite{gerv:19} considered semi-parametric models to decompose the intensity and the log-intensity functions, respectively. 

Our proposal is a nonparametric approach that directly targets the repeated densities, which are modeled as realizations of a  density-valued random element, and the discrete observations from each subpopulation are modeled as i.i.d. realizations from the corresponding density. 
Information is pooled over all subpopulations in a two-step fashion.
First, we identify low-dimensional structures to parsimoniously approximate the densities using data from subpopulations densely sampled by applying FPCA to the centralized log-transformed densities that lie in an unconstrained $L^2$ space.
When transformed back to densities, the low-dimensional structures make up exponential families, where the sufficient statistics are the eigenfunctions of the transformed densities, and 
each density element is compactly represented by its natural parameter or, equivalently, moment coordinate \citep{amar:16}. 
Second, a new and potentially sparsely observed subpopulation is fitted within the approximating family, obtaining maximum likelihood estimate (MLE), and the coordinate estimate is then shrunk towards the population using distributional information, further enforcing information borrowing.
Two shrinkage estimates, namely maximum \emph{a posteriori} (MAP) and best linear unbiased prediction (BLUP), are proposed based on the natural parameter and moment coordinate, respectively.

To illustrate the idea, suppose that random densities are generated from the normal distribution family $\calP = \cbrk{p_{\mu, \sigma} : \mu \in \dsR, \sigma > 0}$ with densities $p_{\mu, \sigma}(x) \propto \exp\left(x (\mu / \sigma^2) - x^2 / (2\sigma^2) \right)$ truncated to $x \in [-1, 1]$.
Observe that a random element $p$ taking values in $\calP$ satisfies $\log p \subset \spn\cbrk{1, x, x^2}$ and thus lies in a low-dimensional space.
The low-dimensional structure can be recovered by applying dimension reduction techniques such as FPCA on $\log p$. 
When inferring the distribution of a subpopulation, information borrowing is performed within the low-dimensional space and is analogous to that for the classical linear mixed effects model.
Remarkably, the proposed methods estimate the low-dimensional structure and pool information fully nonparametrically without imposing parametric shape assumptions.
A worked simulation example is included in \autoref{supp-example:rand.intercept} of the Supplementary Materials. 

The proposed methods combine the flexibility from non-parametric density estimate and the statistical efficiency and interpretability offered by parametric inference and information borrowing.
Ultimately, fitting a density with sparse observations follows likelihood principles in an estimated exponential family, so the required amount of observations by our proposed methods could be substantially fewer than what would otherwise be required for non-parametric estimates.
Moreover, the exponential family approximation has advantages in computation and interpretability.
The approximating family increases in flexibility as its dimension increases and can approximate an arbitrarily complex density.
Rates of convergence of the proposed MLE under the Kullback--Leibler (KL) divergence and $L^1$ norm are established under the asymptotics that the number of training subpopulations, the sample size of each subpopulation, and the sample size of the new subpopulation diverge to infinity.
This result holds for general infinite-dimensional random densities under smoothness  and boundedness conditions.
If the random density actually lies in a finite-dimensional exponential family, then the proposed method converges under KL divergence to the truth with a near-parametric rate.
If the sample size in the new subpopulation is fixed, the proposed MLE converges to the MLE in the best approximating exponential family.

The proposed models and estimation are introduced in \autoref{sec:method}.
Numerical properties of the proposed methods are investigated and compared to alternative parametric and nonparametric density estimates in \autoref{sec:simulation}. 
Data illustrations with age-at-admission of ICU patients and rainfall data in the Haihe river basin are included in \autoref{sec:real.data}. 
Theoretical results are stated in \autoref{sec:theory}. Additional simulation, discussion, and detailed proofs are included in the Supplemental Materials.

\section{Models and Methods}\label{sec:method}

\subsection{Data Setup and Goals}

Let $\cbrk{X_{ij}: j = 1, \dots, N_i, i = 1, \dots, n}$ be the observations of a continuous variable in subpopulation $i = 1, \dots, n$; e.g. the age-at-admission to the Intensive Care Unit (ICU) as described in \autoref{sec:mimic} grouped by primary diagnosis. 
Here $X_{i1}, \dots, X_{iN_i}$ are i.i.d. observations from the $i$th subpopulation with density $p_i$.
Densities $p_1, \dots, p_n$ are modeled as i.i.d. realizations of a random density $p$ taking values in family $\calP$, where 
$\calP$ consists of densities supported on a common compact interval $\calT \subset \dsR$. 
For example, there may exist latent characteristics that determine the susceptibility age profiles of different diseases, so the age distribution of a disease can be regarded the realization of a density-valued random element.
To emphasis the grouping structure, we refer to $p_i$ as the density of the $i$th \emph{subpopulation}, while the distribution of the random density $p$ as the \emph{population}.

Our goal is to estimate the density $p_0$ in a subpopulation with a few i.i.d. observations $X_{0j}$, $j = 1, \ldots, N$, where the number of observations $N$ can potentially be small.
For example, this problem occurs when one want to estimate the age distribution for an uncommon disease. 
We will first obtain low-dimensional structures of $\calP$ from subpopulations where ample observations are available, e.g. from records for more common conditions, and apply the structure to estimate the density in the subpopulation with a small sample size.

\subsection{Proposed Repeated Density Estimation}

\subsubsection{Constructing Approximating Family}

We propose to approximate distribution family $\calP$ with low rank exponential families by utilizing a functional data analytic approach. 
The scenario where the distribution of the random density $p$ is given is discussed here, and the scenario where the distribution is not directly available is described in \autoref{sssec:presmooth}.

The collection of densities form an infinite-dimensional constraint manifold instead of a Hilbert space.
To borrow from the rich literature studying function data as objects in a Hilbert space, we follow a transformation approach \citep{pete:16} and analyze the transformed densities as Hilbertian elements. 
The centralized log-transformation \citep{hron:16} $\clog: \calP \rightarrow L^2(\calT)$ maps a positive density $q \in \calP$ into the $L^2$ space, obtaining
\begin{equation}\label{def:clog}
(\clog q)(t) = \log q(t) - \frac{1}{\size \calT}\int_\calT \log q(x) dx,\, t \in \calT.
\end{equation}

To derive low rank approximations to the densities, we apply the Karhunen--Loève expansion to the transformed trajectories $f = \clog p$ in $L^2(\calT)$, obtaining
\begin{equation}\label{eq:kl.expan.logp}
f(t) = (\clog p)(t) = \mu(t) + \sum_{k = 1}^\infty \eta_k \vphi_k(t), t \in \calT,
\end{equation}
where $\mu(t) = \E f(t)$ is the mean function, $\eta_k = \int_\calT (f(t) - \mu(t)) \vphi_k(t) dt$ is the $k$th component score, and $(\lambda_k, \vphi_k)$ is the $k$th eigenvalue--eigenfunction pair of the covariance function $G(s, t) = \cov(f(s) - \mu(s), f(t) - \mu(t))$ satisfying $\lambda_k \vphi_k(t) = \int_\calT G(s, t)\vphi_k(s) ds$ on $t \in \calT$, for $k=1,2,\dots$. The eigenfunctions are orthonormal, and the eigenvalues are non-negative and non-increasing, satisfying $\lambda_1 \geq \lambda_2 \geq \dots \geq 0$.

Low rank approximating density families to $\calP$ are then obtained by the back-transformation $\clog\inv$ and take the form of exponential families. 
For $K = 1, 2, \dots$, the $K$-dimensional approximation $\calP_K$ to $\calP$ is the exponential family
\begin{gather}
\calP_K  = \cbrk{p_\theta : \theta \in \dsR^K, B_K(\theta) < \infty} \quad \text{ with}\label{def:calpk}\\
p_\theta(t)  = \exp\left(\mu(t) + \sum_{k=1}^{K} \theta_k \vphi_k(t) - B_K(\theta)\right),\; t\in \calT, \label{def:ptheta}
\end{gather}
where
$
B_K(\theta) = \log\left(\int_{\calT}\exp\left(\mu(t) + \sum_{k = 1}^K \theta_k \vphi_k(t) \right)dt \right)
$
is the normalizing constant.
Family $\calP_K$ is identifiable as will be shown in \autoref{supp-prop:identifiable} of supplement. 
The random density $p$ is thus approached by its truncated version 
\begin{equation} \label{eq:pK}
p_{K}(t) \propto \exp(\mu(t) + \sum_{k \leq K} \eta_{k} \vphi_k(t)),\, t\in\calT
\end{equation}
lying in $\calP_K$.

The model components in \eqref{def:ptheta} enjoy interpretation within an exponential family:
$\theta_k$ is the natural parameter which compactly describes the density, the eigenfunctions $\vphi_k$ is the sufficient statistic, 
and mean $\mu$ is the baseline measure.
Analogous to the case for FPCA, the leading eigenfunctions $\vphi_k$, $k=1,\dots,K$ encode the principal modes of variation in the log-densities, and the $\eta_k$ are the scores explaining the most variation;
$\calP_K$ provides the most parsimonious $K$-dimensional description of the random (log-)densities.
Hence, performing density estimation within $\calP_K$ will effectively borrow information from the typical shapes of the densities displayed in the leading eigenfunctions or sufficient statistic.


\subsubsection{Fitting Densities within Approximating Families}

Given a small sample $X_{0j}$, $j = 1, \ldots, N$ from a new subpopulation, an estimate for the underlying density $p_0$ is then obtained as a best fit within the approximating family $\calP_K$.
To assess the goodness of fit, various approaches are available from information theory \citep[e.g.][]{amar:00, amar:16}, and we adopt the information loss as quantified by 
the Kullback--Leibler (KL) divergence, which is defined for two positive densities $q_1$, $q_2$ on $\calT$ as
\begin{equation}\label{def:kldiv}
	\kldiv{q_1}{q_2} = \int_\calT q_1(t) \log\left(
		\frac{q_1(t)}{q_2(t)}
	\right)dt.
\end{equation}
KL divergence $\kldiv{q_1}{q_2}$ from $q_1$ to $q_2$ quantifies the information loss when one approximates $q_1$ by $q_2$. 
To estimate $p_0$, we propose to use the best approximating element in $\calP_K$ minimizing the KL divergence from the empirical distribution.
This leads to the maximum likelihood estimate (MLE) defined by $p_{\rng \theta} \in \calP_K$, where
\begin{align}\label{def:fpca.mle.pop}
    \rng \theta &= \argmax_{\theta: B_K(\theta) < \infty}\left( \theta^T \bar\vphi_0- B_K(\theta) \right),
\end{align}
and $\bar \vphi_0 \in \dsR^K$ is the sufficient statistic constructed from the sample $X_{0j}$, in which the $k$th element is $N\inv \sum_{j \leq N} \vphi_k(X_{0j})$, $k = 1, \ldots, K$.

\subsubsection{Incorporating Population-level Information}

In addition to the typical shapes of random densities, we borrow information from the distribution of these random densities in the population
to reduce variation when fitting a new density.
This will result in a significant improvement especially if the sample size in the new subpopulation is small.
%

The distribution of $p$ is captured by those of the component scores $\eta_k$, $k=1,\dots, K$. 
Incorporating such distributional information into the likelihood function, we propose maximum \emph{a posteriori} (MAP) estimate $\rng p_{\text{MAP}} = p_{\rng \theta_{\text{MAP}}} \in \calP_K$, where the parameter estimate $\rng \theta_{\text{MAP}}$ is defined as
\begin{align}\label{def:fpca.map.pop}
\rng \theta_{\text{MAP}} &= \argmax_{\theta : B_K(\theta) < \infty} \left(\theta^T \bar \vphi_0 - B_K(\theta) + \log \pi(\theta)\right),
\end{align}
and $\pi$ is the unconditional density of the $\eta_k$. 
Analogous to the PACE procedure \citep{yao:05}, $\pi$ is constructed using independent marginal distributions being normal distributions with mean $\E \eta_k=0$ and variance $\Var\eta_k=\lambda_k$, $k=1,\dots,K$.
Here $\rng \theta_{\text{MAP}}$ is a shrinkage of the MLE $\rng \theta$ towards the population $\pi$. 

To accommodate non-normal unconditional distributions for the $\eta_k$, we propose a second approach to incorporate the population-level knowledge by working with the moment parameterization. 
The moment parameter, also referred to as the moment coordinate, for an element $p_\theta$ in exponential family $\calP_K$ is $\xi = (\xi_1, \dots, \xi_K) \in \dsR^K$ where $\mco_k = \int_\calT \vphi_k(t) p_\theta(t) dt$. 
The moment parameter is in one-to-one correspondence with natural parameters \citep[see, e.g., ][]{amar:16}, 
noting
\begin{equation}\label{eq:m.e.translate}
	\mco = \partial_\theta B_K(\theta).
\end{equation}
For brevity, we use $\mco(\cdot): \theta \mapsto \mco(\theta)$ and $\theta(\cdot): \mco \mapsto \theta(\mco)$ to denote the one-to-one mappings derived from \eqref{eq:m.e.translate}.

Motivated by the fact that the MLE of the natural parameter is the method of moments estimate under moment parameterization, we consider shrinkage estimate based on the the moment parametrization. 
Let $\eta_0$ be the natural parameter of the truncated random densities $p_{0, K} \in \calP_K$ of $p_0$ according to \eqref{eq:pK} and $\tau_0 = \mco(\eta_0) \in \dsR^K$ the corresponding moment coordinate. 
We propose to estimate the moment parameter $\tau_0$ by the best linear unbiased predictor (BLUP) of $\tau_0$ given the sample sufficient statistic $\bar \vphi_0$, defined as 
\begin{align}\label{def:fpca.blup.pop}
    \rng \mco_{\text{BLUP}} &= \blup(\tau_0 \mid \bar \vphi_0) \defas \cov(\tau_0, \bar \vphi_0) \var(\bar \vphi_0)\inv (\bar \vphi_0 - \E\bar \vphi_0) + \E\tau_0 \nonumber\\
    &= \var(\tau_0) \var(\bar \vphi_0)\inv (\bar \vphi_0 - \E \tau_0) + \E\tau_0,
\end{align}
where the last equality is due to $E (\bar \vphi_0 \mid \tau_0) = \tau_0$ and double expectation.
Note that here the covariance and expectation are computed under the joint distribution of $(\tau_0, \bar\vphi_0)$. 
We thus obtain a plug-in estimate $\rng\theta_{\text{BLUP}} = \theta(\rng \mco_{\text{BLUP}})$ of the natural parameter.

The MAP \eqref{def:fpca.map.pop} and BLUP \eqref{def:fpca.blup.pop} estimates not only borrow typical shapes of densities via the sufficient statistic $\vphi$ analogous to the MLE \eqref{def:fpca.mle.pop}, but also take advantage of the distribution of the random density through the component scores. 
We refer to the information embedded in the component scores of log-densities as population-level information. 
The MAP \eqref{def:fpca.map.pop} and BLUP \eqref{def:fpca.blup.pop} are, respectively, a shrinkage of the MLE \eqref{def:fpca.mle.pop} estimate towards the mean under the natural and the moment coordinate.
As will be demonstrated in our numerical investigation later, this extra shrinkage is precious for reducing information loss especially given small samples.

\subsubsection{Pre-smoothing and Unknown Distribution of $p$}
\label{sssec:presmooth}

Oftentimes, the random densities $p_1,\dots,p_n$ and their distribution are unavailable but only discrete samples $\cbrk{X_{ij}: j = 1, \dots N_i; i = 1 \dots, n}$ are accessible. 
In this case, we utilize \emph{pre-smoothing} to construct pilot density estimates for the samples, analogous to the approach taken by \cite{pete:16, han:20}.
While many density estimators may be applied, for example kernel density estimate (KDE) and logspline \citep{koop:91}, for theoretical consideration we adopt the weighted KDE 
\begin{equation}\label{def:weighted.KDE}
\ck p_i(t) = 
\left(
\int_\calT 
\sum_{j = 1}^{N_i} \knl\left(\frac{x - X_{ij}}{h}\right) w(x, h) dx
\right)\inv
%
\sum_{j = 1}^{N_i} \knl\left(\frac{t - X_{ij}}{h}\right) w(t, h), \; t \in \calT
\end{equation}
to handle boundary bias, where $\knl$ is a kernel function and $w(t, h) = 1 / \int_{(t-b)/h}^{(t-a)/h} \knl(u) du $ is the weight for $\calT = [a, b]$. 
In the pre-smoothing step, the pilot density estimator is applied on each sample to obtain pre-smoothed densities $\ck p_1, \dots, \ck p_n$. 

FPCA is then performed on the transformed densities $\ck f_i = \clog \ck p_i$, $i = 1, \dots, n$ to obtain sample mean  $\hmu(t) = n\inv \sum_{i=1}^{n} \ck f_i(t)$, sample covariance function $\h G(s, t) = n\inv \sum_{i=1}^{n} (\ck f_i(s) - \h\mu(s))(\ck f_i(t) - \h\mu(t))$, the associated eigenvalues  $\hl_k$ and eigenfunctions $\hphi_k$, and the component scores $\h\eta_i \in \dsR^K$ in which the $k$th element is $\h\eta_{ik} = \inner{\ck f_i - \h\mu}{\h\vphi_k}$ for $k = 1, \dots, n - 1$. 
A sample version of the approximating families $\h\calP_K$ is obtained as
\begin{gather}
\h\calP_K = \cbrk{\h p_\theta : \theta \in \dsR^K, \h B_K(\theta) < \infty}, \text{ with} \label{def:calpk.presmooth}\\
\h p_\theta(t) = \exp\left(\h\mu(t) + \sum_{k = 1}^K \theta_k \h\vphi_k(t) - \h B_K(\theta)\right),\, t \in \calT,
\end{gather}
where $\h B_K(\theta) = \log\left(
\int_\calT \exp\left(\h\mu(t) + \sum_{k \leq K} \theta_k \h\vphi_k(t) dt\right)
\right)$ is the normalizing constant.
Inference for new samples are conducted based on $\h \calP_K$ with sufficient statistic $\h\vphi_{0,\ave} \in \dsR^K$, in which the $k$th element is $N\inv \sum_{j = 1}^N \h\vphi_k(X_{0j})$; to avoid double over-scripts here we use subscript ``ave'' to denote the average. The distribution of random densities can be accessed via the empirical distribution of the sample component scores. 
The proposed estimators in fully sample-based versions are then defined using the estimated quantities.
The MLE and MAP within $\h\calP_K$ are $\h p_{\h \theta_{\text{MLE}}}$ and $\h p_{\h \theta_{\text{MAP}}}$, respectively, where
\begin{align}
\h \theta_{\text{MLE}} &= \argmax_{\theta: \h B_K(\theta) < \infty}\left( \theta^T \h\vphi_{0,\ave}- \h B_K(\theta) \right), \label{def:fpca.mle}\\
\h \theta_{\text{MAP}} &= \argmax_{\theta : \h B_K(\theta) < \infty} \left(\theta^T \h\vphi_{0,\ave} - \h B_K(\theta) + \log \h\pi(\theta)\right), \label{def:fpca.map}
\end{align}
in parallel with  \eqref{def:fpca.mle.pop} and \eqref{def:fpca.map.pop},
and $\h\pi$ is the product of normal densities with mean
$0$ and variance $s^2_{\h\eta, k}$, the sample variance of $\cbrk{\h\eta_{ik}, i = 1, \dots, n}$, for $k = 1, \dots, K$.

For BLUP, 
$\E\tau_0$ and $\var(\tau_0)$ in \eqref{def:fpca.blup.pop} are estimated by the sample mean $\h \tau_\ave = n\inv \sum_{i \leq n}\h\tau_i$ and covariance matrix $\Sigma_{\h\tau}$ of moment coordinates $\cbrk{\h\tau_i = \mco(\h\eta_i) : i = 1, \dots, n}$ in $\h\calP_K$.
Further decompose $\var(\bar\vphi_0) = \E\var(\bar\vphi_0 \mid \tau_0) + \var(\tau_0)$, and estimate the first term via integration by
\begin{equation*}
\Sigma_{\h\vphi_\ave} = \frac{1}{nN} \sum_{i = 1}^n
\int_\calT \left(\h\vphi(t) - \h\tau_i\right) \left(\h\vphi(t) - \h\tau_i\right)^T \ck p_i(t) dt,
\end{equation*}
where $\h\vphi: t \mapsto (\h\vphi_1(t), \dots, \h\vphi_K(t)) \in \dsR^K$ is the sufficient statistic. 
In combine, the plug-in BLUP estimate for the moment parameter within $\h\calP_K$ is obtained as 
\begin{equation}\label{def:fpca.blup}
\h \xi_{\text{BLUP}} = \Sigma_{\h\tau} \left(\Sigma_{\h\vphi_\ave}  + \Sigma_{\h\tau} \right) ^{-1} (\h \vphi_{0, \ave} - \h\tau_\ave) + \h\tau_\ave,
\end{equation}
and the associated natural parameter $\h\theta_{\text{BLUP}} = \h\theta(\h \mco_{\text{BLUP}})$ is the solution to $\h\mco_{\text{BLUP}} = \partial_\theta \h B_K(\theta)$.

When the random densities $p_1, \ldots, p_n$ are observed in their entirety, a scenario considered for the interest of theory and referred to as the subpopulations being \emph{completely observed}, it suffices to skip the pre-smoothing step and proceed with the observed densities.
In this case, FPCA is performed on the transformed trajectories $f_i = \clog p_i$, $i = 1, \dots, n$ to obtain sample mean and covariance as $\tl \mu(t) = n\inv \sum_{i \leq n} f_i(t)$ and $\tl G(s, t) = n\inv \sum_{i \leq n} (f_i(s) - \tl\mu(s)) (f_i(t) - \tl\mu(t))$, eigenvalues and eigenfunctions $\tl\lambda_k$ and $\tl\vphi_k$, and sample component scores $\tl\eta_{ik} = \inner{f_i - \tl \mu}{\tl\vphi_k}$, $k = 1, \dots, n - 1$. 
The approximating family $\tl\calP_K$ with completely observed densities is defined as
\begin{gather}
\tl \calP_K = \cbrk{\tl p_\theta : \theta\in\dsR^K, \tl B_K(\theta) < \infty}, \text{ with}\label{def:calpk.comp}  \\
\tl p_\theta(t) = \exp\left(\tl\mu(t) + \sum_{k = 1}^K \theta_k \tl\vphi_k(t) - \tl B_K(\theta)\right),\, t \in \calT
\end{gather}
for $K=1,2,\dots$, where $\tl B_K(\theta) = \log \int_\calT \exp\left(\tl\mu(t) + \sum_{k \leq K} \theta_k \tl\vphi_k(t) dt\right)$ is the normalizing constant.
The proposed MLE, MAP and BLUP are then defined analogously to that within $\h\calP_K$.
We use tilde overscripts to denote estimated quantities with fully observed densities.

We refer to the process of constructing approximating exponential families from the discrete observations from $p_1,\dots, p_n$ as the \emph{training} step, which emphasizes that the $K$-dimensional families $\h \calP_K$ are flexibly derived from the data despite themselves being parametric families compactly described by a few natural parameters. 
The random densities $p_1, \ldots, p_n$ are referred to as the training densities associated with the training subpopulations, and the corresponding samples as training samples.
For performance assessment, we call the process constructing density estimates from new observations within the approximating families the \emph{fitting} or \emph{testing} step, and correspondingly, define  testing subpopulation, density, and sample.

The number of components $K$ controls the flexibility of $\h \calP_K$ and is the key tuning parameter trading off bias and variance when fitting a density.
We propose to select $K$ using a penalized approach by considering the Akaike information criterion (AIC), and set $K=K^*$ which minimizes $\aic(K) \defas 2 K - 2\log \h p_K$, where $\h p_K$ is the density estimates within $\h\calP_K$ using either \eqref{def:fpca.mle}, \eqref{def:fpca.map}, or \eqref{def:fpca.blup}. 
In practice, we observe that the optimal $K^*$ often increases as the fitting sample size increases.
Intuitively, the proposed component selection automatically utilizes a more flexible model given a larger sample, in which case the additional variance that comes with the flexibility becomes more affordable. 

The proposed procedure is computationally efficient, since the proposed MLE, MAP, and BLUP estimates are performed within exponential families, so the optimization tasks in obtaining  \eqref{def:fpca.mle}--\eqref{def:fpca.blup} are convex, and the discrete samples are succinctly summarized using the sufficient statistics.

\section{Simulation Studies}
\label{sec:simulation}

We study the numerical properties for the proposed methods and alternative parametric and non-parametric density estimators under various scenarios. 
Independent training samples $\cbrk{X_{ij}: j = 1, \ldots, N_i}$ were drawn from densities $p_i$, $i=1,\dots,n$, respectively, where the $p_i$ are independent copies of a random density $p$. 
Independent testing samples $\cbrk{X^*_{lj}: j = 1, \ldots, N_l^*}$ were then drawn from additional independent copies $p_l^*$ of $p$ for $l = 1, \dots, n^*$. 
Typically, training sample sizes were larger, e.g. $N_i = 200$, while testing samples were smaller, e.g., $N_l^* = 10$.
The small sample sizes in the testing samples pose a challenge for the density estimators compared.
The underlying family $\calP$ and the distribution of $p \in \calP$ varied between scenarios. The experiment under each scenario was repeated 500 times.

The proposed methods \eqref{def:fpca.mle}, \eqref{def:fpca.map}, and \eqref{def:fpca.blup} are referred to as FPCA\_MLE, FPCA\_MAP, and FPCA\_BLUP. 
Our proposals were compared against the maximum likelihood estimate within the ground truth family $\calP$ as an idealized parametric approach, referred to as MLE; non-parametric methods including KDE \citep{shea:91} and logspline \citep{koop:91}; and a repeated point processes (RPP) approach proposed by \cite{gerv:19}.
These methods do not require training.

The mean KL divergence for the density estimates was used to assess the overall estimation performance. 
For density estimate $\h p^*_l$ for the $l$th testing sample, information loss was evaluated by KL divergence $\kldiv{p^*_l}{\h p^*_l}$ defined in \eqref{def:kldiv}, and we examined the average 
$$
\text{MKL} = \frac{1}{n^*} \sum_{l=1}^{n^*} \kldiv{p^*_l}{\h p^*_l}
$$
over $n^*$ testing samples to quantify performance.
Smaller MKL indicates better estimation performance. 
Results for an alternative error metric, the mean integrated squared errors (MISE), and additional simulation cases are included in the Supplementary Materials.

\subsection{Flexible Exponential Families}\label{ssec:simu.exponential}

We first set the density families to be either 
the truncated normal or the bimodal distributions.
For the truncated normal scenario, each sample was generated from $N(\mu, \sigma^2)$ truncated on $[-3, 3]$ with $\mu \sim \text{Unif}(-2, 2)$ and $\sigma \sim \text{Unif}(2, 4)$.
For the bimodal scenario, each sample was generated according to density
$p(x) \propto \exp(
(4 + \theta) x - (26.5 + \theta) x^2
+ 47 x^3 - 25 x^4
), x\in [0, 1]$ 
with $\theta \sim \text{Unif}(0, 10)$.
In each Monte Carlo experiment, $n = 50$ training samples of size $N_i = 200$ and $n^* = 100$ testing samples with sizes $N=10$ or 50 were created.

The mean KL divergence in \autoref{fig:bim.trnml.boxplot} shows that the proposed FPCA\_MLE using either KDE as pre-smoother or with the complete training densities observed both produce estimates comparable to MLE without requiring explicit specification of a parametric family.
Among all the non-parametric density estimators, the proposed methods are shown to be stable and best-performing, especially when the testing samples are small, in which case the proposed methods outperformed MLE by borrowing strength among different samples, whereas MLE does not pool information.
In particular, for truncated normal family with $N^* = 10$, on average the information loss of the proposed MAP and BLUP were around 20\% less compared to that of MLE within ground truth family, and more than 40\% less compared to that of KDE.

\begin{figure}[!h]
	\centering
	\includegraphics[width=1\textwidth]{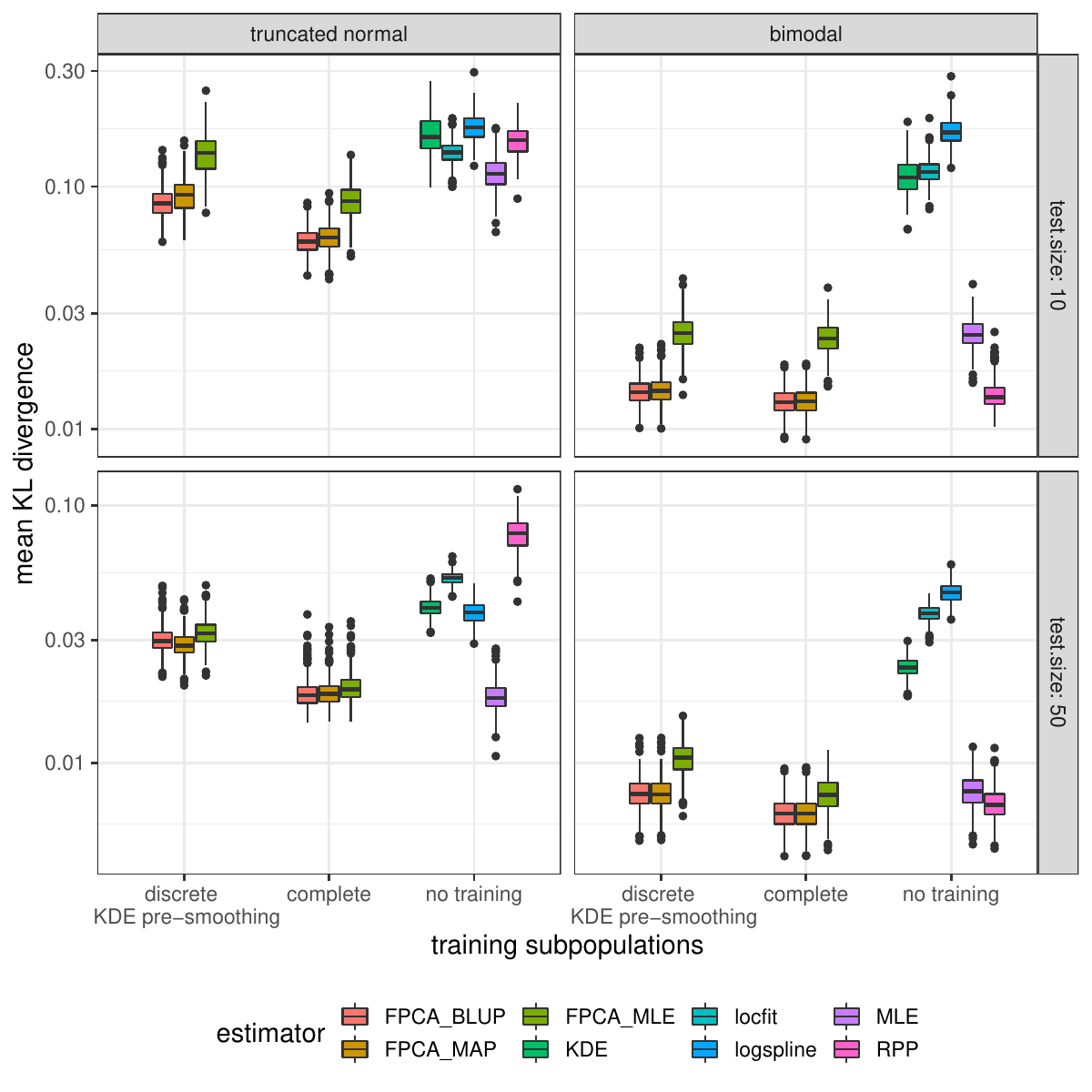}
	\spacingset{1}
	\caption{\textit{Boxplots of mean KL divergence. RPP, repeated point processes approach \citep{gerv:19}; locfit, local polynomial density estimate \citep{load:96}; logspline, adaptive logspline \citep{koop:91}. Both discrete training samples of sizes $N_i = 200$ with KDE pre-smoothing and completely observed densities are considered for our proposed approaches, namely FPCA\_MLE, FPCA\_MAP, and FPCA\_BLUP.
	}}
	\label{fig:bim.trnml.boxplot}
\end{figure}


\subsection{Gaussian Mixture Family}\label{example:mixGauss}

Contrary to the previous scenario where the underlying families were exponential families, 
we next estimate densities from a family consisting of mixtures of Gaussian distributions, which forms a multi-modal non-exponential family.
This is to further demonstrate the flexibility of the proposed methods to a ``mis-specified'' case and also demonstrate the data-adaptive selection of the number of components. 

The samples were generated following a mixture of three Gaussian distributions with density $p(x) \propto \sum_{l = 1}^{3} \theta_l \phi((x - \mu_l) / \sigma_l)$ truncated to $x \in [-3, 3]$, where $\phi$ is the standard normal density and the mixture probabilities $(\theta_1, \theta_2, \theta_3)$ were generated from $\text{Dirichlet}(1/3, 1/3, 1/3)$. For $l = 1, 2, 3$, we generated the mixture mean $\mu_l$ from $\text{Unif}(-5, 5)$ and standard deviation $\sigma_l$ from $\text{Unif}(0.5, 5)$.
Here, instead of performing a direct MLE which is known to be unstable for mixture distributions, the expectation--maximization (EM) algorithm in the variant of \cite{lee:12} is used as the baseline parametric density estimator.

The left panel of \autoref{fig:gauss.mixture} compares the relative error of various density estimates to that of the EM algorithm as a baseline reference.
The proposed methods enjoyed the least information loss for all but the largest sample size $N^*=500$ when EM method became the best. 
When the testing sample size $N^* = 25$, on average, the information loss of proposed MAP and BLUP were 20\% less than that of RPP, and they were all at least five times better compared to the EM algorithm.
Remarkably, this is achieved by only utilizing the first few components on average, as demonstrated in the right panel of \autoref{fig:gauss.mixture}.
This shows that the proposed methods are effective alternatives to the EM algorithm, which overcome local maximum issues by performing convex optimization within approximating exponential families.
In addition, by using the proposed AIC method for component selection, the proposed estimates utilized more components as the size of fitting sample increased, suggesting that the AIC selection adapts to the increased complexity of the underlying family as unveiled by a larger sample. 
Indeed, the larger the sample size, the more representative the sample will be for the underlying distribution, which has 8 free parameters in our setup. 
By using more components, the proposed methods can accommodate the increased complexity shown by the data such as multi-modality, which may otherwise be hidden when fitting with fewer components.


\begin{figure}[!h]
	\includegraphics[width=0.95\textwidth]{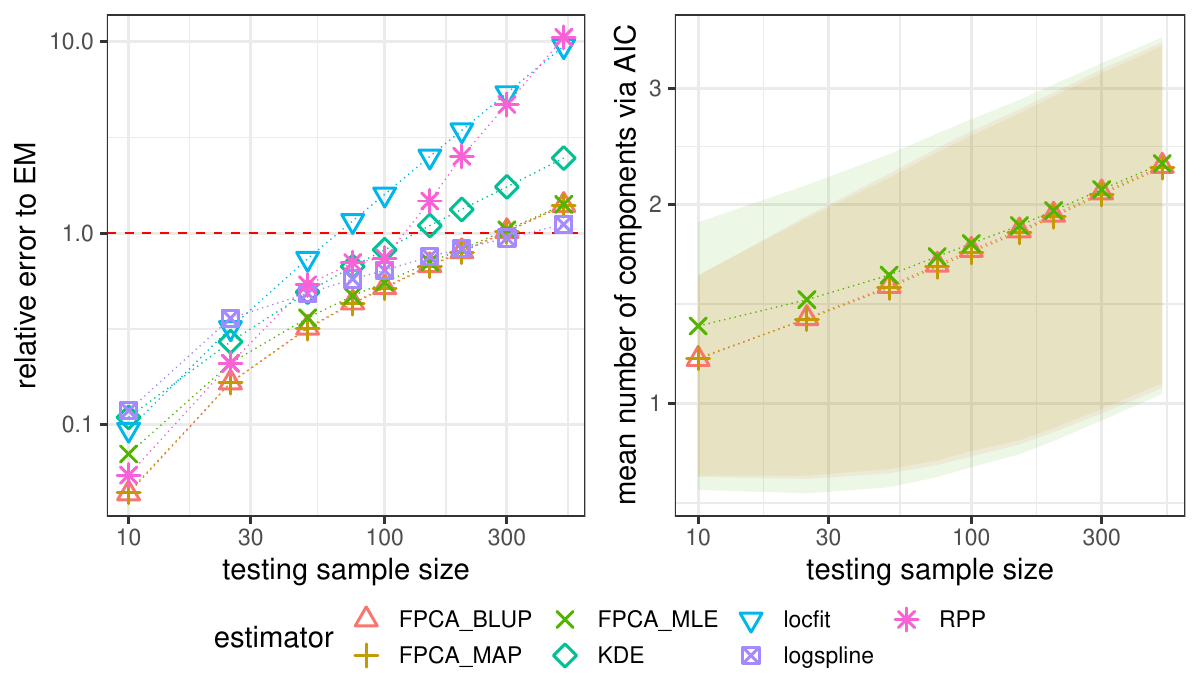}
	\spacingset{1}\caption{\textit{
	The left panel shows the ratios of the average mean KL divergence of the compared methods over that of the EM algorithm, where a ratio smaller than 1 indicates better performance than EM. 
	The right panel displays the average number of components $K$ selected by the proposed AIC method, with ribbons indicating one standard deviation around the means. 
	Axes in both panels are drawn in $\log10$ scale. 
	Training samples had sample sizes $N_i=50$, and logspline was used as the pre-smoother. 
	RPP, repeated point processes approach \citep{gerv:19}; locfit, local polynomial density estimate \citep{load:96}; logspline, adaptive logspline \citep{koop:91}; FPCA\_MLE, FPCA\_MAP, and FPCA\_BLUP, the proposed methods in the MLE, MAP, and BLUP variants.
}}
	\label{fig:gauss.mixture}
\end{figure}


\section{Real Data Applications}\label{sec:real.data}

\subsection{MIMIC Data}\label{sec:mimic}

We consider the age-at-admission distributions of patients with critical conditions in the Medical Information Mart for Intensive Care (MIMIC) data \citep{john:16}, collected from patients admitted to critical care units at a large hospital in Boston between 2001 and 2012. 
The number of patients under different primary diagnosis are highly different.
Ample data points are available to establish a precise non-parametric density estimate for common diseases, as shown in the left panels of \autoref{fig:mimic.subpop}, whereas observations for uncommon conditions are scarce and cannot afford flexible non-parametric approaches such as KDE, as displayed in right panels of \autoref{fig:mimic.subpop}. Estimating the densities in the hard-to-observe subpopulations of patients with uncommon conditions is the emphasis in our analysis.

\begin{figure}[!h]
	\centering
	\includegraphics[width=1\textwidth]{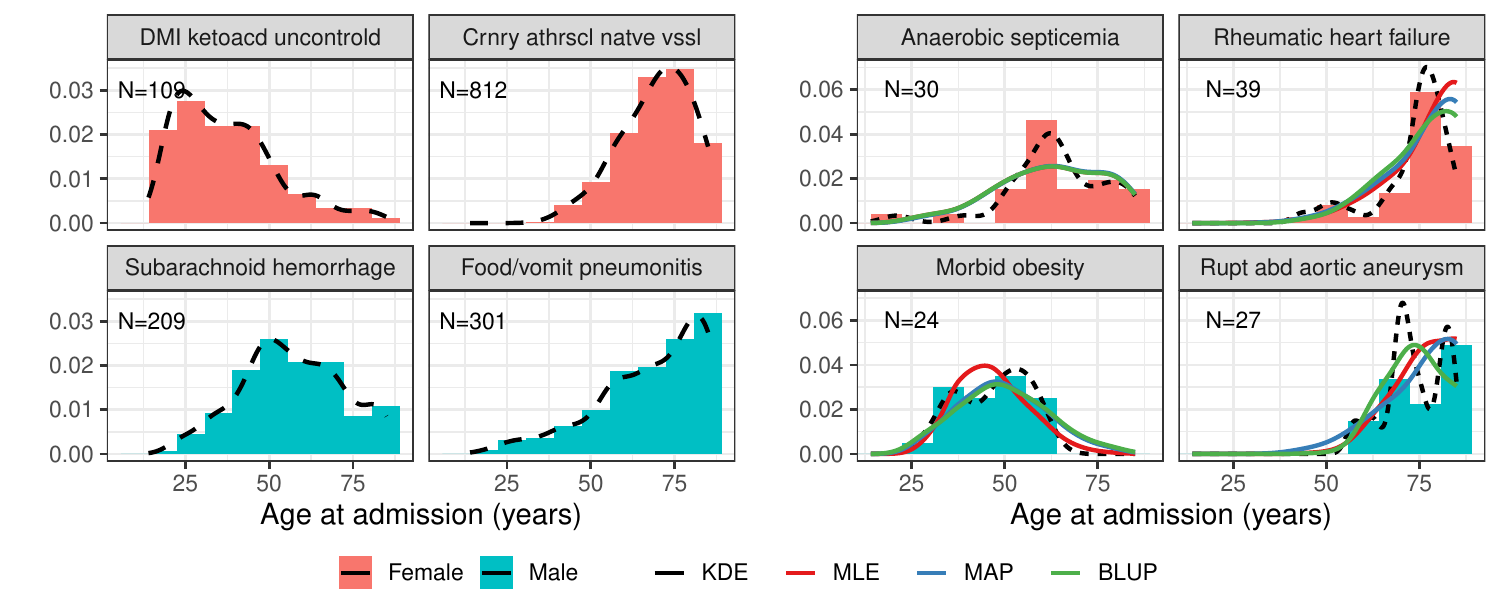}
	\spacingset{1}\caption{\textit{Each panel shows the age-at-admission distribution of either male or female patients with a different primary diagnosis, defining a subpopulation. 
		Four training subpopulations with ample observations were shown in the left panels, and four testing subpopulations with scarce observations are shown in the right panels. 
		Along with the KDE, the proposed density estimates MLE, MAP, and BLUP are displayed for the testing subpopulations. KDE, kernel density estimate; MLE, MAP and BLUP, the proposed density estimates.
	}}
	\label{fig:mimic.subpop}
\end{figure}

Gender and diagnosis code (ICD9) combinations were used to define the subpopulations, where we excluded birth-related, unclassified, or unspecified diagnosis. 
A total of $940$ subpopulations with more than 5 observations each were analyzed, among which $75$ subpopulations with at least $75$ observations were used as the training while the rest as testing. 
For example, the upper left panel in \autoref{fig:mimic.subpop} shows a training sample formed by $109$ female patients with juvenile type diabetes with ketoacidosis (DMI ketoacd uncontrold), while the bottom right panel shows a testing sample consisting of 27 male patients with abdominal aortic aneurysm rupture (Rupt abd aortic aneurysm).
We focused on estimating the densities supported in the age range of $[14, 85]$. 
KDE was utilized as the pre-smoother and the bandwidth was set to the median of the bandwidths in the training samples selected by the method of \cite{shea:91}. 

As illustrated in right panels of \autoref{fig:mimic.subpop}, the proposed density estimates produced sensible results given small test samples, which were smoother than the estimates given by KDE and avoided multiple spurious modes.  

The modes of variation in the density functions obtained by varying each of the first three natural parameters of $\h\calP_K$ are shown in the left panels of \autoref{fig:mimic.shrink.mov}.
The first three modes of variations correspond to a shift in the age distribution towards the old, the elderly, and the middle age, respectively. This suggests the existence of latent disease characteristics that affect the age profiles of different conditions.

The effects of shrinkage in the proposed MAP and BLUP estimators are demonstrated in the right panels of \autoref{fig:mimic.shrink.mov} for the four rare conditions displayed in the right panels of \autoref{fig:mimic.subpop}.
Subpopulations with smaller sample sizes were typically shrunk more towards the population mean using the MAP and BLUP methods that pool population information from the training samples, in which case the estimate differed more notably from the MLE within the approximating family.

To numerically assess the density estimates, we consider cross-entropy which also targets the information loss instead of the KL divergence, since the underlying densities are unknown.
Given any density estimator $\h p$, we approach the cross-entropy $\cxEtrp(p, \hat p) \defas - \E_p \log{\h p}$, which is the quantity in $\kldiv{p}{\h p}$ that depends on the estimate $\h p$, by an unbiased leave-one-out estimate
\begin{equation}\label{def:looce}
\cxEtrp_{\text{LOO}}(p, \h p) \defas -\frac{1}{N} \sum_{j \leq N} \log{\h p_{-j}(X_j)},
\end{equation}
where $\h p_{-j}$ is a density estimate constructed with all but the $j$th observations. 
A summary of cross-entropy applied on 844 testing samples using the proposed methods and KDE are reported in \autoref{tbl:mimic.looce}, where a smaller cross-entropy indicates a better estimate. Occasionally KDE resulted in non-finite cross entropy, leading to a smaller number of samples reported.

\autoref{tbl:mimic.looce} shows that the proposed methods produced density estimates that better reflect sample information compared to KDE by having smaller mean cross-entropy, and also smaller standard deviations. 
To quantify the sensitivity w.r.t. the training subpopulations, we partitioned the training samples into 5 folds, held out one fold at a time in the training step, fitted the densities for the testing samples, and reported the range of the five mean cross-entropies for a sensitivity measure. 
The proposed methods are shown to be insensitive to varying training subpopulations, which suggests that the selected training samples are sufficient for our methods to capture the major structure of the underlying family. 

\begin{figure}[!h]
	\includegraphics[width=1\textwidth]{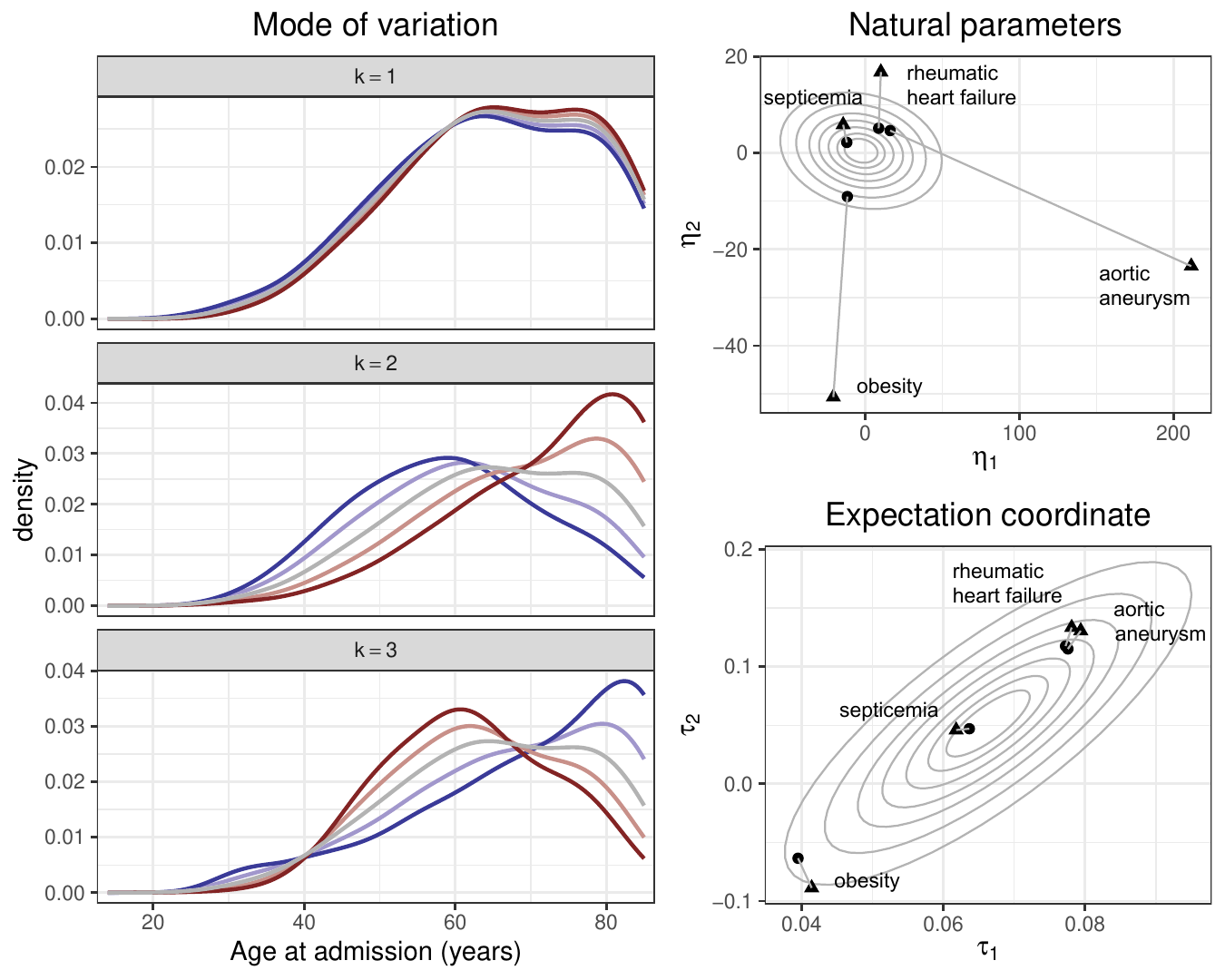}
	\spacingset{1}\caption{\textit{
		Left: mode of variation by varying each of the first three natural parameters of the estimated family $\h\calP_K$, where blue and red stand for, respectively, lower and higher values.
		Right: for the four sparsely observed subpopulations shown in the right panels of \autoref{fig:mimic.subpop}, the MLE parameter estimates (triangles) are shrunk in the natural parameter space (upper panel) and the expectation coordinate space (lower panel) w.r.t. the population distribution of training component scores (elliptical contours), producing the corresponding MAP (dots, upper right) and BLUP (dots, lower right) estimates. 
	}}
	\label{fig:mimic.shrink.mov}
\end{figure}

\begin{table}[!h]
\begin{center}
	\caption{Cross-entropy \eqref{def:looce} of different estimators applied on the MIMIC data.}
	\label{tbl:mimic.looce}
	
\begin{tabular}{lrlrrr}
	\toprule
	size & $n^*$ & estimator & mean & sd & sensitivity\\
	\midrule
	&  & MLE & 4.60 & 1.62 & 0.062\\
	&  & MAP & 4.13 & 0.42 & 0.013\\
	& \multirow{-3}{*}{\raggedleft\arraybackslash 423} & BLUP & 4.11 & 0.36 & 0.039\\
	\cmidrule{2-6}
	\multirow{-4}{*}{\raggedright\arraybackslash (0,10]} & 402 & KDE & 5.56 & 2.42 & -\\
	\cmidrule{1-6}
	&  & MLE & 4.03 & 0.32 & 0.015\\
	&  & MAP & 3.97 & 0.25 & 0.005\\
	& \multirow{-3}{*}{\raggedleft\arraybackslash 343} & BLUP & 3.97 & 0.24 & 0.017\\
	\cmidrule{3-6}
	\cmidrule{2-2}
	\multirow{-4}{*}{\raggedright\arraybackslash (10,35]} & 341 & KDE & 4.16 & 0.37 & -\\
	\cmidrule{1-6}
	&  & MLE & 3.93 & 0.22 & 0.013\\
	&  & MAP & 3.92 & 0.21 & 0.004\\
	&  & BLUP & 3.92 & 0.21 & 0.008\\
	\cmidrule{3-6}
	\multirow{-4}{*}{\raggedright\arraybackslash (35,75]} & \multirow{-4}{*}{\raggedleft\arraybackslash 78} & KDE & 3.99 & 0.22 & -\\
	\bottomrule
\end{tabular}

\end{center}

\spacingset{1}\textit{
The results are stratified by size of testing samples, where the number of testing samples involved is listed as $n^*$. 
The sensitivity of proposed methods are the range of the mean cross-entropy across 5 experiments, each holding out one of five folds in the training samples.
}
\end{table}


\subsection{Precipitation Return Level Estimation} \label{sec:haihe}

We then analyze the precipitation data in the Haihe region to illustrate the versatility of the proposed methods in providing good estimates for both the body part of the density function and the tail probability. 
We demonstrate that our proposed methods can borrow information from stations with longer records to construct approximate distribution families for annual maxima of hourly rainfall, which lead to better estimates of densities and return levels for those stations having shorter records.

Predicting the likelihood of extreme precipitation is of great importance because of the serious damage such events can cause to economy and social life. We are in particular interested in heavy rainfall during short time period. 
The likelihood of extreme precipitation event in a certain region can be quantified using the $T$-year return level, which is defined as the level expected to be exceeded once in $T$ years, or equivalently the $1 - 1/T$ quantile of the marginal distribution under stationarity and independence assumptions. 
Sample quantile is a straight forward estimate for small $T$, which is referred to as the empirical return level, but this approach requires more than $T$ data points being available. Otherwise, distribution models are necessary for estimating the return level for longer return periods. 

We focus on estimating the 5- to 30-year return levels at various locations in the Haihe River basin, a vital region in northern China that includes mega-cities such as Beijing and Tianjin. Precipitation data from 232 weather stations around Haihe River basin from 1961 to 2012 are available from the National Meteorological Information Center (NMIC) and the China Meteorological Administration (CMA) and are used for this analysis. 
All the 232 stations investigated have at least 30 years of records. 

Our target is the distribution of the annual maximum rainfall per hour, where each station defines a subpopulation.
As the distribution of precipitation is heavily right-skewed, the pre-smoothing density estimates could be zero within the working domain, creating an issue when we calculate the log-density. To address this issue and the positive constraint, log-precipitation is used as the working variable.
See \autoref{supp-appendix:log.inf} of the Supplementary Materials for detail. 

We compared return levels estimation and the information loss of the proposed methods with standard approaches. 
To imitate the situation of newer weather stations having shorter records, we randomly selected $n = 50$ stations as training sites and retained the full records in these stations for training, and another $n^* = 100$ testing sites where records from randomly chosen 10 consecutive years were used for fitting.
Estimates of density and 5- to 30-year return levels were obtained from non-parametric approaches including our proposed methods and standard KDE, as well as parametric approaches of maximum likelihood estimates within the families of Gamma and GEV distribution.
The performance metric used to evaluate density estimation was cross-entropy \eqref{def:looce},
which measures the overall fit. 
For return level estimates, we used the relative difference of the estimates to the empirical return levels that are obtained from full records in the testing stations,
which measures the fit in the right tail. 
The entire process was repeated $10$ times, and the overall average cross-entropies and the relative differences of the return levels averaged over the stations are reported in \autoref{fig:haihe.return.level}.

The proposed methods in the MLE, MAP, and BLUP variants are all significantly better than the KDE and GEV approaches under the information loss. The estimates using Gamma family of distributions were comparable to our results but underperformed the MAP and BLUP methods which borrow information from the population. For return level estimation, 
the proposed MAP and BLUP methods had the least variation compared to MLE, Gamma and GEV methods, where the last three methods do not pool information. For shorter (5- and 10- year) return levels, the bias of all methods were comparable, but for longer (20- and 30- year) return levels, the bias of our methods were significantly smaller than Gamma and GEV methods.  
In summary, our proposed methods have the best performance in estimating the return levels as well. 
Hence, the proposed methods provide better fits compared to classical approaches in both the bulk and the tail parts of the distributions and lead to better quantification of interesting characteristics of the distribution such as return levels from a small number of discrete observations.

These numerical results come as a positive surprise, as we only used 10 years of records at the testing sites and yet were able to obtain good estimates for 20- and even 30-years return levels.
The proposed approaches are able to estimate the tail structure from the data rather than relying on assumptions, and borrow information from the population enforced through parameter shrinkage.

\begin{figure}[!h]
	\centering
	\includegraphics[width=0.975\textwidth]{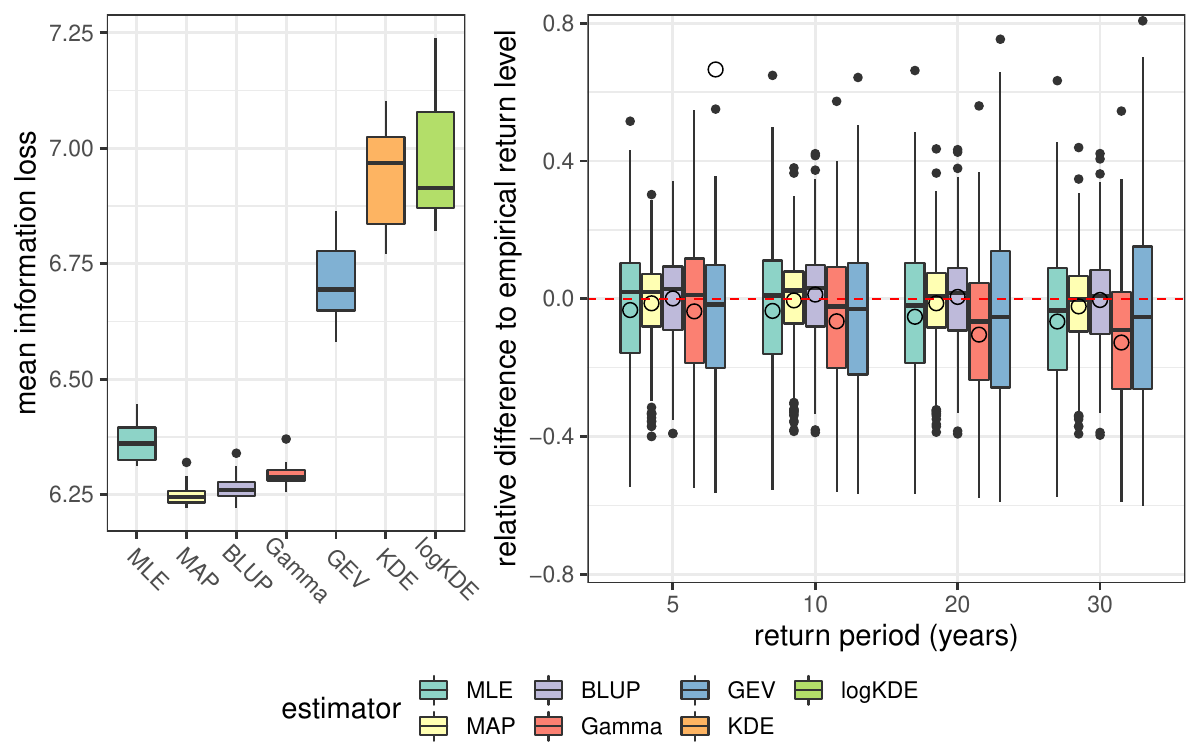}
	\spacingset{1}\caption{\textit{Left panel: Cross-entropies of different estimators applied on the Haihe rainfall data. Left panel: Cross-entropies computed by \eqref{def:looce}, where each method was fitted with decennial records over $10$ repeated experiments. MLE, MAP, and BLUP, the proposed methods; KDE, kernel density estimate; Gamma, density fit within a Gamma distribution; GEV, density fit within a GEV family.
	Right panel: Relative differences of estimated return levels using decennial records compared to the empirical ones using full records ($\geq 30$ years) for the testing weather stations. Each solid dot represents the average relative difference at a station averaged over Monte Carlo experiments, while each  hollow circle represents the mean over all stations. Estimates closer to the horizontal zero-reference line are better. The GEV had some extreme over-estimate outliers, pulling its mean errors outside plotting range. }}
	\label{fig:haihe.return.level}
\end{figure}


\section{Theoretical Results}\label{sec:theory}

We state the theoretical results for the proposed methods including the identifiability of the approximating families, consistency of MLE within those families, and consistency of the component scores. 
While the major motivation and justification for the proposed methods are the small sample performance as demonstrated in the simulations and real data applications, 
in this section we also establish large sample properties of the proposed estimators to justify their broad validity.
\autoref{ssec:consistencyLarge} and \autoref{ssec:consistencyComp} investigate the large testing sample scenario where the testing sample size $N$, the number of training samples $n$, and the training sample sizes $N_i$ all diverge, while \autoref{ssec:consistencySmall} considers the small testing sample case when $N$ is small and kept fixed but $n$ and $N_i$ diverge.
We present theoretical properties of the MLE as justification for the proposed methods, since for a large testing sample the BLUP and MAP estimates emphasize the observation but not the population-level information and will be close to ``degenerate'' at the MLE estimate. 
Proofs and additional discussions are included in the Supplementary Materials.

For theoretical development, densities in the underlying family are required to be smooth and uniformly bounded away from zero and infinity. This condition is standard when analyzing densities as functional data given discrete observations; see, e.g., \cite{pete:16,han:20}.

\begin{enumerate}[label=(A\arabic*)]
	\item \label{ass:boundedpdf} All densities $q \in \calP$ are continuously differentiable. Moreover, there is a constant $M > 1$ such that, for all $q \in \calP$, $\norminf{q}, \norminf{1/q}$ and $\norminf{q^\prime}$ are bounded above by $M$.
\end{enumerate}

Under \autoref{ass:boundedpdf}, the approximating families $\calP_K$, $\tl\calP_K$, and $\h\calP_K$ as defined in \eqref{def:calpk}, \eqref{def:calpk.comp}, and \eqref{def:calpk.presmooth} are identifiable, thanks to the fact that the entries of the sufficient statistic are orthogonal to each other and to the constant function, as shown in \autoref{supp-prop:identifiable} in the Supplemental Materials.

\subsection{Consistency of MLE as $N$ Diverges} \label{ssec:consistencyLarge}
The proposed MLE estimates within low-dimensional approximating family $\h\calP_K$ are shown to be consistent.
The MLE is computed with an additional i.i.d. sample $X_1, \dots, X_N$ of size $N$, of which the density $p$ is an independent copy of the random density. The approximating family $\h\calP_K$ is obtained from $n$ discretely observed training subpopulations each with $N_i \ge m$ observations. Theoretical results for this practical scenario are presented here, and results for MLE obtained within $\calP_K$ and $\tl\calP_K$ are included in the Supplementary Materials.
For brevity, we use $\h p = \h p_{N, K} \in \h\calP_K$ to denote the MLE density estimate with parameter $\h \theta = \h\theta_{N, K}$ defined in \eqref{def:fpca.mle}.

The consistency results rely on the following conditions imposed on the covariance of the log-transformed random density.

\begin{enumerate}[label=(C\arabic*)]
	\item \label{ass:consist.kl.pop} For process $f = \clog p \in L^2(\calT)$, the eigenvalues and eigenfunctions of the integral operator $\calG$ associated with the covariance function $G$ satisfy, as $J \toinf$,
	\begin{align*}
	\left( \sum_{k = 1}^J \| \vphi_k^\prime\|_\infty \right) \left(\sum_{k > J} \lambda_k\right) &= O(1), \\
	\left( \sum_{k = 1}^J \| \vphi_k^\prime\|_\infty \right)^{2/3} \left(\sum_{k>J}\lambda_k\right) &= o(1).
	\end{align*}
	
	\item \label{ass:cov.G.eigen.sep} For all $k = 1, \dots, \rnk \calG$, the non-zero eigenvalues are distinct and satisfy $\delta_k \defas \inf_{l\not=k} \abs{\lambda_k - \lambda_l} / 2 > 0$; the estimated eigenfunctions are properly aligned so that  $\inner{\vphi_k}{\tl\vphi_k} \geq 0$ and $\inner{\vphi_k}{\h\vphi_k} \geq 0$.
\end{enumerate}

Condition \autoref{ass:consist.kl.pop} ensures that the high frequency eigenfunctions are not overly bumpy relative to the amount of information they are carrying as quantified by the eigenvalues, so that the trajectories can be reasonably represented with the first few components. 
Condition \autoref{ass:cov.G.eigen.sep} is standard for simplifying presentation. We assume $\rnk \calG = \infty$ in discussion hereafter unless otherwise stated.

Conditions \autoref{ass:boundedpdf}, \autoref{ass:consist.kl.pop}, and \autoref{ass:cov.G.eigen.sep} in combine are sufficient for the consistency of MLE within $\calP_K$ and $\tl\calP_K$ as shown in \autoref{supp-thm:kldiv.all} of the Supplementary Materials. 
When discrete training samples and pre-smoothing are involved, in parallel to those in \cite{pete:16}, additional conditions \autoref{ass:presmooth.MISE}, \autoref{ass:presmooth.UnifTight}, and \autoref{ass:presmooth.superTight} are also required as listed in the \hyperref[appendix:presmooth]{Appendix}. Those conditions basically requires that the pilot density estimator performs well uniformly, provided the sample sizes are sufficiently large. Here  \autoref{ass:presmooth.MISE} is a uniform MISE requirement, \autoref{ass:presmooth.UnifTight} requires the $L^\infty$ error to be $O_p(a_m)$, and \autoref{ass:presmooth.superTight} is a requirement for the rate of sample sizes as the number of samples increases, so that one can derive uniform rates of convergence for the pre-smoothed training densities and log-densities under the $L^\infty$ norm.
The conditions \autoref{ass:presmooth.MISE}, \autoref{ass:presmooth.UnifTight}, and \autoref{ass:presmooth.superTight} can be satisfied by the weighted kernel density estimate (KDE) in \eqref{def:weighted.KDE} with proper bandwidths and kernels, e.g., $h \asymp m^{-1/3}$ and the Gaussian kernel. See \autoref{supp-sec:presmooth} of the Supplementary Materials for detail.

\subsubsection{Main Consistent Results}

We have the following result for the proposed MLE within $\h\calP_K$. Recall that $n$ is the number of training subpopulations, $m$ is the minimal training sample size in \autoref{ass:presmooth.superTight}, $N$ is the size of the fitting sample, and sequences $\cbrk{a_m}$, $\cbrk{b_m}$ are those in \autoref{ass:presmooth.MISE} and \autoref{ass:presmooth.UnifTight}.

\begin{thm}\label{thm:kldiv.sample}
	Under \autoref{ass:boundedpdf}, 
	\autoref{ass:consist.kl.pop}, \autoref{ass:cov.G.eigen.sep}, 
	\autoref{ass:presmooth.MISE}, \autoref{ass:presmooth.UnifTight}, and \autoref{ass:presmooth.superTight},
	if $n$, $N$, $K$ increase to infinity in such a way that 
	\begin{align*}
	A_K \sqrt{K/N} \tozero, \quad
	K A_K (1 / \sqrt{n} + b_m) \tozero, \quad
	\sum_{k\leq K} \lambda_k \inv (1 / \sqrt n + a_m) = O(1), \\
	A_K \sum_{k\leq K} \delta_k\inv (1/\sqrt{n} + b_m) \tozero, \qquad
	\sum_{k\leq K} (\delta_k \lambda_k) \inv (1/\sqrt{n} + b_m) = O(1),
	\end{align*}
	where
	\begin{equation}
	A_K = 2 \max \left(
	\size \calT ^{-1/2},
	(\sum_{k = 1}^K \norminf{\vphi_k'})^{1/3}
	\right),
	\end{equation}
	then 
	\begin{align*}
	\kldiv{p}{\h p_{N, K}} &=  
	O_p\left(
		\left(\frac{1}{n} + b_m^2\right) \left(
		K + 
		\sum_{k\leq K} \delta_k\inv
		\right)^2
		+ \sum_{k>K}\lambda_k + \frac{K}{N}
	\right).
	\end{align*}
\end{thm}
\begin{proof}
	See \autoref{supp-thm:kldiv.all} in the Supplementary Materials.
\end{proof}
Note that the result is highly flexible and suits any general infinite-dimensional families. In particular, we do not assume the underlying $\calP$ to be a finite-dimensional exponential family.
The three summands in the rate correspond to the training error, a bias term due to the approximating family, and a variance term from fitting with $N$ observations. Varying the number of components $K$ trades off errors from these three terms, as a more complex model with a larger $K$ will have larger training and fitting errors but smaller approximation bias.
Quantity $A_K$ is closely related to \autoref{ass:consist.kl.pop} and the equivalence of $L^\infty$ and $L^2$ norm in the eigenspaces of covariance operator; see \autoref{supp-prop:eigenspace.lipschitz} in the Supplementary Materials.

An immediate corollary is the convergence of the densities in the $L^1$ norm by Pinsker's inequality \citep{wain:19}, which states that the KL divergence dominates squares of $L^1$ norm of differences.
\begin{corollary}
	Under the conditions of \autoref{thm:kldiv.sample}, 
\begin{equation*}
\normp[1]{\h p_{N, K} - p}^2 =  
	O_p\left(
	\left(\frac{1}{n} + b_m^2\right) \left(
	K + 
	\sum_{k\leq K} \delta_k\inv
	\right)^2
	+ \sum_{k>K}\lambda_k + \frac{K}{N}
	\right).
\end{equation*}
\end{corollary}

Our general result directly implies the next corollary where the underlying family is actually a $K_0$-dimensional exponential family for some finite $K_0$ with no approximation bias resulting from truncating the log-densities. Note that \autoref{ass:consist.kl.pop} is not required, and by convention we set $\lambda_k$, $\vphi_k$ to be zero for $k > K_0$.

\begin{corollary}\label{coro:kldiv.finiteK}
	Under \autoref{ass:boundedpdf}, 
	\autoref{ass:cov.G.eigen.sep},
	\autoref{ass:presmooth.MISE}, \autoref{ass:presmooth.UnifTight}, and \autoref{ass:presmooth.superTight}, if $\calP$ is a $K_0$-dimensional exponential family for some finite $K_0$, then for any fixed $K \geq K_0$,
	\begin{align*}
	\kldiv{p}{\h p_{N, K}} 
	&= O_p\left(
		\frac{1}{n} + b_m^2	+ \frac{1}{N}
	\right),
	\end{align*}
	as $n, N \toinf$ and $m = m(n) \toinf$.
\end{corollary}
\begin{proof}
	See \autoref{supp-coro:kldiv.finiteK} in the Supplementary Materials.
\end{proof}

\subsubsection{Selecting $K$}

We provide an example when the requirements in \autoref{thm:kldiv.sample} are satisfied in the following proposition. 

\begin{ps}\label{prop:optimal.K}
	Under \autoref{ass:cov.G.eigen.sep}, if 
	\begin{enumerate}[label = (\arabic*)]
		\item the eigenfunctions coincide with the polynomial or trigonometric basis so that $\norminf{\vphi_k'} \asymp k$ while $k$ large;
		
		\item eigenvalues have polynomial decay: $\lambda_k \asymp k^{-r}$ with $r \geq 3$ as $k \toinf$;
		
		\item minimum sizes of the $n$ training sample satisfies $m(n) = n^{r_m}$ for some $r_m > 0$;
		
		\item weighted KDE \eqref{def:weighted.KDE} is used as pre-smoother with $h \asymp m^{-1/3} = n^{-r_m / 3}$ as  $n\toinf$; and
		
		\item for any fixed $0 < r_a < 1/6$, the number of components $K$ used for fitting in dependence on $r, n, N, r_m$ and $r_a$ satisfies
		$$
		K = \h K^* \asymp \min\left(
		n^{\frac{1}{4(r+1)}}, n^{\frac{r_a r_m}{r+1}}, N^{1/r}
		\right) \text{ as } n, N \toinf,
		$$
	\end{enumerate}
	then the conditions for \autoref{thm:kldiv.sample} are satisfied with $a_m = n^{-r_a r_m}$, $b_m = n^{-r_m/3}$.
	As $n, N \toinf$,
	$$
	\kldiv{p}{\h p_{N, \h K^*}} = O_p\left(n^{\gamma_1} + N^{-1 + 1/r}\right),
	$$ 
	where $\gamma_1 = r_a r_m (1-r) / (r + 1)$ if $r_m \leq 1/(4 r_a)$, 
	and $\gamma_1 = -1/4 + 1 / (2 r + 2)$ if $r_m > 1/(4 r_a)$.
\end{ps}
\begin{proof}
	See \autoref{supp-prop:optimal.K.all} in the Supplementary Materials.
\end{proof}

Such $\h K^*$ is derived by finding the optimal and feasible $K$ for the dominating terms in the convergence rate.
Compared to the result for optimal $K$ without presmoothing as in \autoref{supp-prop:optimal.K.all} of the Supplementary Materials, it is suggested that in this example, pre-smoothing does not slow down the convergence rate for the proposed MLE if the training subpopulations are densely observed so that $r_m$ is large, e.g. $r_m > 3/2$. Since if $r_m > 1/(4 r_a)$, then $\kldiv{p}{\tl p_{N, \tl K^*}} \asymp \kldiv{p}{\h p_{N, \h K^*}}$ with an optimal $\tl K^*$ for MLE within $\tl\calP_K$.

This result is particularly interesting since it suggests that, for better convergence rate, when pre-smoothing with KDE, slight over-smoothing may be preferred. Indeed, noting $\gamma_1 < 0$ will be smaller if $r_a r_m > 1/4$, it is desirable to have larger $r_a$ so that smaller $r_m$ can satisfy $r_a r_m > 1/4$, which means smaller training samples can be considered ``dense''.
As required in \autoref{supp-prop:preS.ass.good}, the only restrictions for $r_a$ are $0 < r_h + r_a < 1/2$ and $0 < r_a < r_h$, and hence if we choose a smaller $r_h > 1/4$, a larger $r_a$ would be possible, making $r_a r_m > 1/4$ easier to satisfy.

\subsection{Consistency of Component Scores} \label{ssec:consistencyComp}

Another interesting question is the convergence of the component scores of pre-smoothed training trajectories $\ck f_i = \clog \ck p_i$, $i = 1, \ldots, n$, which are used to capture distribution of random densities in MAP \eqref{def:fpca.map} and BLUP \eqref{def:fpca.blup}.
Notation-wise, let $\ck f_{i,K} = \h\mu + \sum_{k \leq K} \ck\eta_{ik} \h\vphi_k$ be the truncated pre-smoothed log-density
with $\ck\eta_{ik} = \inner{\ck f_i - \h\mu}{\h\vphi_k}$ for $k = 1, \ldots, K$, $i = 1, \dots, n$. 
The next proposition shows that the estimated component scores $\ck \eta_{ik}$ converge uniformly to the truth $\eta_{ik} = \inner{f_i - \mu}{\vphi_k}$.

\begin{ps}\label{prop:comp.score}
	Under \autoref{ass:boundedpdf}, \autoref{ass:presmooth.MISE}, \autoref{ass:presmooth.UnifTight}, \autoref{ass:presmooth.superTight},
	and \autoref{ass:cov.G.eigen.sep},
	for any $K < \infty$, as $n\toinf$, we have
	\begin{equation*}
	\sup_{i \leq n} \sup_{k \leq K} \abs{\ck\eta_{ik} - \eta_{ik}} = O_p\left(a_m + (n\invsqr + b_m) \sup_{k \leq K} \delta_k\inv\right).
	\end{equation*} 
\end{ps}
\begin{proof}
	See \autoref{supp-prop:comp.score.all} of the Supplementary Materials.
\end{proof}

With the previous proposition, it is easy to see $\sup_{i \leq n} \sup_{k \leq K} \abs{\ck\eta_{ik} - \eta_{ik}} = O_p(a_m + n\invsqr)$ as $n \toinf$ if $K$ is fixed. 
To have a similar result with $K \toinf$ as well, we need $n, K \toinf$ in such a way that $(n\invsqr + b_m) \sup_{k \leq K} \delta_k\inv \tozero$. Under the conditions of \autoref{prop:optimal.K}, we can set $K \asymp n^{\beta}$ for any $\beta < \max(1/2, r_m/3) / (1 + r)$, then the sample component scores are consistent uniformly. In particular, $K = \h K^*$ satisfies the requirement.

\subsection{Consistency of MLE in a Small and Fixed Sample} \label{ssec:consistencySmall}

The following result further characterizes the asymptotic behavior of the proposed MLE for a given new fitting sample with a small and fixed sample size $N$, where $K$ is finite and fixed and there is a large amount of observations from training subpopulations.

\begin{ps}\label{prop:rate.training.thetahat}
	Under \autoref{ass:boundedpdf}, \autoref{ass:presmooth.MISE}, \autoref{ass:presmooth.UnifTight}, \autoref{ass:presmooth.superTight}, and a fixed $K$, for a given sample $X_1, \dots, X_N$ drawn from density $p \in \calP$, assume the MLE $\rng p_{N,K} \in \calP_K$ exists and is unique with corresponding parameters $\rng \theta \in\dsR^K$.
	Then the MLE $\h p_{N,K} \in \h\calP_K$ exists and is unique with probability tending to 1 as $n\toinf$, where the $n$ is the number of training subpopulations used to construct $\h\calP_K$. The corresponding parameter $\h\theta_n \in\dsR^K$ converges to $\rng\theta$ in probability with
\begin{align*}
	\normp{\h\theta_n - \rng\theta} &= O_p(1/\sqrt n + a_m + b_m), \\
	\kldiv{\rng p_{N, K}}{\h p_{N,K}} &= O_p\left(
		\frac{1}{\sqrt n} + a_m + b_m
	\right).
\end{align*}
\end{ps}
\begin{proof}
	See \autoref{supp-prop:rate.training.thetahat} and \autoref{supp-coro:rate.training.kldiv} of in the Supplemental Materials.
\end{proof}

\appendix
	\appendix

\section{Assumptions for Pre-smoothing}\label{appendix:presmooth}

When discrete training samples and pre-smoothing are involved, the following conditions in parallel to (D1), (D2), and (S1) in \cite{pete:16} are also required.
Let $q \in \calP$ be a fix density;
the conditions here are imposed on family $\calP$ but not on the distribution of the random density $p$.
Further, let $\ck q$ to be the corresponding pre-smoothed density, e.g. according to the weighted KDE defined in \eqref{def:weighted.KDE}.

\begin{enumerate}[label=(D\arabic*)]
	\item \label{ass:presmooth.MISE} There exists a positive sequence $b_N = o(1)$ as $N \toinf$, such that based on an i.i.d. sample of size $N$ from density $q \in \calP$, the density estimator $\ck q$  satisfies $\ck q \geq 0$, $\int_\calT \ck q = 1$ and $$\sup_{q \in \calP} \E(\normp[2]{\ck q - q}^2) = O(b_N^2).$$
	
	\item \label{ass:presmooth.UnifTight} There exists a positive sequence $a_N = o(1)$ as $N \toinf$ and some $R > 0$ such that based on an i.i.d. sample of size $N$ from a population for which the density is $q \in \calP$, density estimator $\ck q$ satisfies $\ck q \geq 0$, $\int_\calT \ck q = 1$ and 
	$$
	\sup_{q \in \calP} \Pr(\norminf{\ck q - q} > R a_N) \rightarrow 0 \text{ as } N \toinf.
	$$
\end{enumerate}    

\begin{enumerate}[label=(S\arabic*)]    
	\item \label{ass:presmooth.superTight} Let $\ck q$ be a density estimator satisfying \ref{ass:presmooth.UnifTight} and $N_i = N_i(n), i = 1, \ldots, n$ be the sample sizes in the training subpopulations. There exists a sequence of lower bounds $m(n) \leq \min_{1 \leq i \leq n} N_i$ such that $m(n) \to \infty$ as $n \toinf$ and
	$$
	n \sup_{q \in \calP} \Pr(\norminf{\ck q - q} > R a_m) \to 0,
	$$
	where $\ck q$ is obtained with sample size $N(n) \geq m(n)$.
\end{enumerate}


\

\

\begin{center}
	{\large\bf SUPPLEMENTARY MATERIAL}
\end{center}
The Supplemental Materials include additional simulation results and details of the theoretical works.

\bibliographystyle{chicago}
\bibliography{mr_ref.bib}

\end{document}